%% file: main.tex
\documentclass[%
 aip,
 amsmath,amssymb,
 reprint,%
]{revtex4-1}

\usepackage{graphicx}
\usepackage{dcolumn}
\usepackage{bm}

\usepackage[utf8]{inputenc}
\usepackage[T1]{fontenc}
\usepackage{etoolbox}

\input{commands}

\usepackage{lineno,hyperref}
\usepackage{MnSymbol}
\usepackage{mathtools}
\usepackage{amsfonts}
\usepackage{yfonts}
\usepackage{color}
\usepackage{amsthm}

\newcommand{\stressFlux}{u \dot{\wedge} \mathcal{T}}
\newtheorem{proposition}{Proposition}
\newtheorem{remark}{Remark}

\begin{document}

\title{A differential geometric description of thermodynamics in continuum mechanics with application to Fourier-Navier-Stokes fluids}
\author{F. Califano}
\email{f.califano@utwente.nl}
 \affiliation{ 
Robotics and Mechatronics Department, University of Twente, 7522 NH Enschede, The Netherlands
}%
\author{R. Rashad}%
\affiliation{ 
Robotics and Mechatronics Department, University of Twente, 7522 NH Enschede, The Netherlands
}%

\author{S. Stramigioli}
\affiliation{ 
Robotics and Mechatronics Department, University of Twente, 7522 NH Enschede, The Netherlands
}%

\begin{abstract}
A description of thermodynamics for continuum mechanical systems is presented in the coordinate-free language of exterior calculus. First, a careful description of the mathematical tools that are needed to formulate the relevant conservation laws is given. Second, following an axiomatic approach, the two thermodynamic principles will be described, leading to a consistent description of entropy creation mechanisms on manifolds. Third, a specialisation to Fourier-Navier-Stokes fluids will be carried through. 
\end{abstract}

\maketitle

\section{Introduction}

In this paper we present a comprehensive description of thermodynamics in a differential geometric setting for a general non relativistic continuum mechanical system. In order to deal with manifestly coordinate-invariant equations valid on general Riemannian manifolds, we make use of a formulation based on exterior calculus. We show how to complement this covariant formulation, already well known in continuum mechanics\cite{Kanso2007OnMechanics}, with a consistent description of thermodynamics valid in general mechanical field theories, generalising the procedure carried through on Euclidean spaces using vector calculus\cite{Asinari2016OverviewDynamics}.

The conceptual steps to derive an entropy equation valid in non-equilibrium thermodynamic conditions are generalised at a manifold level by:
\begin{enumerate}
    \item clarifying the differential geometric nature of the thermodynamic variables (e.g., intensive vs extensive variables) as tensor fields living on manifolds; 
    \item reviewing and referring to the mathematical tools used to describe continuum mechanics on manifolds and to formulate coordinate-free physical conservation laws, like conservation of energy (the first principle);
    \item axiomatically postulating a covariant version of the Gibbs relation on manifolds, which leads to a valid entropy equation for any continuum.
\end{enumerate}

As final contribution we specialise the equations to Fourier-Navier-Stokes fluids, in order to shed light on the entropy creation mechanisms in Newtonian fluids on manifolds.

As part of the background material used in this work, we will compactly present and point out the main properties of covariant field equations of continuum mechanics, whose treatment is distilled from the foundational references\cite{Marsden,Frankel2011TheIntroduction,Marsden1970HamiltonianHydrodynamics,Marsden1984ReductionAlgebras,Marsden1984SemidirectMechanics,Morrison1998,Kanso2007OnMechanics}, which inspired more recent works (e.g., \cite{rashad2021port,rashad2021portb,Gilbert2019AMechanics,Mora2021OnFluids}) aiming at investigating specific geometric structure underlying the equations.

We stress that in this work we will work out covariant field equations of continuum mechanics combined with thermodynamics, without investigating the geometric structure of the resulting model. For example we will not investigate the thermodynamic state space as infinite-dimensional extension of the contact geometry peculiar of finite-dimensional thermodynamic systems, extensively studied in \cite{vanderSchaft2018GeometryProcesses} and references therein. Such an extension is currently under investigation and should encompass infinite-dimensional GENERIC \cite{Grmela2014ContactDynamics} and dissipative port-Hamiltonian formulations \cite{Badlyan2018OpenSystems}.

Nevertheless any geometric structure underlying the equations should stick with the conservation laws presented in this work axiomatically, and could lead to models which can conveniently lead to structure-preserving spatial discretisation schemes (see e.g., \cite{Arnold2006FiniteApplications,10.1007/0-387-38034-5_5,Gawlik2021AFlow}) and novel control strategies. Furthermore we exclude from this work a functional analytic treatment of the equations, and use differential geometric mathematical tools assuming that the stated integrals are well defined and that the tensor fields are smooth enough to make sense of the used differential operators.

The paper is organised as follows. In Sec. \ref{sec:background} we review and refer to the mathematical tools needed to describe conservation laws on manifolds and revise the formulation of continuum mechanics using the covariant language of exterior calculus. In Sec. \ref{sec:thermodynamics} we give the thermodynamic description of continuum mechanical processes, which are specialised to the Fourier-Navier-Stokes case in Sec. \ref{sec:FNS}. Conclusions are given in Sec. \ref{sec:conclusion}.

\paragraph*{Notation} 

A compact, orientable, $n$-dimensional Riemannian manifold $M$ with (possibly empty) boundary $\partial M$ models the spatial container of a non relativistic continuum mechanical system. It possesses a metric field $g$ and induced Levi-Civita connection $\nabla$ (with chart-dependent Christoffel symbols denoted with $\Gamma_{ij}^{k}$). The space of vector fields on $M$ 
is defined as
the space of sections of the tangent bundle $TM$, that will be denoted by $\vecf{M}$. The space of differential $p$-forms is denoted by $\Omega^p(M)$ and we also refer to $0$-forms as functions, $1$-forms as covector fields and $n$-forms as top-forms. For any
$v\in \vecf{M}$, we use the standard definitions for the interior product by $\iota_v:\Omega^p(M) \to \Omega^{p-1}(M)$ and the Lie derivative operator $\mathcal{L}_v$ acting on tensor fields of any valence. 
The Hodge star operator
$\star:\Omega^p(M) \to \Omega^{n-p}(M)$, the volume form $\calV=\star 1$, as well as the "musical" operators $\flat:\vecf{M} \to \Omega^{1}(M)$ and $\#:\Omega^{1}(M) \to \vecf{M}$, which respectively transform vector fields to 1-forms and vice versa, are all uniquely induced by the Riemannian metric in the standard way. 
When making use of Stokes theorem $\int_M \extd \omega=\int_{\partial M} \textrm{tr}(\omega)$ for $\omega \in \Omega^{n-1}(M)$, the trace operator $\textrm{tr}:=i^*$ is the pullback of the canonical inclusion map $i: \partial M \hookrightarrow M$.
When dealing with tensor-valued forms we adopt the additional convention that a numerical index $i\in\{1,2\}$ in the subscript of a standard operator on differential forms indicates whether the operator acts on the ``first leg'' (the tensor value) or on the ``second leg'' (the underlying form) of the tensor-valued form on the respective side of the operator: For a vector-valued $m$-form $u \tens \beta$, for instance, we define $\star_2(u \tens \beta):=u \tens \star\beta$. Vector fields will be identified with vector-valued $0$-forms and top-forms will be identified with covector-valued $(n-1)$-forms in a non ambiguous way when the operators will require it. We will introduce along the paper further operators and constructions where needed. To complement a purely coordinate-free notation, we will represent tensorial quantities of interest in local coordinates, denoted by $x$, when considered insightful, using Einstein summation convention.

\section{Background}
\label{sec:background}

\subsection{Reynolds transport theorem}
We summarise relevant theorems involving differential geometric formulations of differentiation of integrals, discussing useful insights that are needed in the context of this work. We refer to Chapter 4.3 in\cite{Frankel2011TheIntroduction} for thorough
proofs of the results.

\subsubsection{The time-invariant case}
Consider an $n$-dimensional manifold $M$ and a 1-parameter group of diffeomorphisms $\phi_t: M \to M$
with $t \in \mathbb{R}$. The group property indicates that the diffeomorphisms $\phi_t$ generate a flow, i.e., $\phi_s (\phi_t x) =\phi_{t+s} x$. The diffeomorphisms $\phi_t$ induce a motion on any $p$-dimensional compact submanifold $V$ of $M$ since, in a neighbourhood of $V$, we can define $V(t):=\phi_t V \subseteq M$. Furthermore let $u(x)= \ddt{\phi_t(x)}|_{t=0}
$ be the velocity vector field generating the motion $\phi_t$. Notice that $u\in \vecf{M}$ is defined in such a way to be a time-independent vector field on $M$, due to the assumption of the 1-parameter group property of $\phi_t$. Given a time-independent $p$-form $\alpha \in \Omega^p(M)$, we are interested in understanding the rate of change of integrals of the form $I(t):=\int_{V(t)}\alpha=\int_V \phi_t^* \alpha$, that we indicate with $\dot{I}:=\frac{\extd{} I(t)}{\extd{} t}$. An explicit calculation of the latter shows that
\begin{equation}
 \label{eq:ReynoldsAutonomous0}
    \dot{I}=\int_{V(t)} \lie{u}{\alpha},
\end{equation}
which can be further expressed
using Cartan's formula $\mathcal{L}_u=\extd{\circ \iota_u}+\iota_u \circ\extd{}$, and Stokes theorem as
 \begin{equation}
 \label{eq:ReynoldsAutonomous}
    \dot{I}=\int_{\partial V(t)} \tr{ \iota_u{\alpha}} + \int_{V(t)} \iota_u \extd{\alpha}.
\end{equation}
Notice that in case $M$ is a Riemannian manifold and $\alpha$ is a top-form (i.e., $p=n$) the second term in (\ref{eq:ReynoldsAutonomous}) vanishes and using the identity\cite{Hirani2003DiscreteCalculus} (Lemma 8.2.1 in the reference) $\iota_u \alpha= \star \alpha \wedge \star u^{\flat}$ and distributing the trace over the wedge one obtains
\begin{equation}
    \dot{I}=\int_{\partial V(t)} \tr{\star \alpha} \wedge \tr{\star u^{\flat}}.
\end{equation}
The latter expression is useful to clearly see that $\dot{I}=0$ in case the vector field $u$ is parallel to $\partial V(t)$, which is equivalent to $\tr{\star u^{\flat}}=0$. If additionally we chose $V=M$ one obtains $\dot{I}=0$ always, since the fact that $\phi_t$ is a group of diffeomorphisms of $M$ into itself forces the vector field $u$ to be parallel to $\partial M$ (in case $\partial M \neq \emptyset$).

\subsubsection{The time-varying case}
\label{sec:tvc}

For the physical conservation laws we are interested to formalise, the fields that appear under the integral sign are in general 
time-dependent. This is due to the fact that they will be function of physical states, which
evolve in time according to the dynamic equations of the system. Furthermore the domain of integration may vary with diffeomorphisms $\phi_t$ that do \textit{not} form a flow, i.e., the generating vector field $u$ 
may also in general be
time-dependent. We refer in this case to $\phi_t$ as 1-parameter family (instead of group) of diffeomorphisms. This happens e.g., in fluid mechanics, when the volume $V(t)$ follows a fluid portion. In this case the velocity field $u$ is given by the macroscopic velocity vector field of the fluid, and as such is in general time-dependent (as it is the solution of the unsteady Navier-Stokes equations).  

The expression in this case generalises to
\begin{equation}
 \label{eq:ReynoldsNonAutonomous}
    \dot{I}=\int_{V(t)} \dot{\alpha} + \lie{u}{\alpha},
\end{equation}
where $\alpha$ is now a possibly time-dependent $p$-form, $\dot{\alpha}:=\partial \alpha / \partial t$ and $u(t,x)=\ddt{\phi_t(x)}$ is the possibly time-dependent vector field generating the motion of $V(t)$. The proof of this theorem relies on identifying the relevant time-dependent fields with tensor fields on the space-time product manifold $\mathbb{R} \times M^n$, over which the vector field generates a flow, a necessary condition to implement the calculations leading to (\ref{eq:ReynoldsAutonomous0}), which can then be repeated in this case with a space-time velocity field $\Bar{u}=\frac{\partial}{\partial t}+u$, leading to (\ref{eq:ReynoldsNonAutonomous}). We stress that this result relies on the assumption of a Newtonian space-time, in which the absolute time function $t$ can be used as first coordinate on $\mathbb{R} \times M^n$.

\subsection{Continuum mechanics on Manifolds}

We briefly review the differential geometric formulation of continuum mechanics on an $n$-dimensional Riemannian manifold $M$.
We remark that the choice on the tensorial nature of the variables is not unique, especially in the way \textit{stress} is treated. We will use the description which serves best to the scope of this paper, using a so-called \textit{Eulerian}
formulation, i.e., one that treats state variables as tensor fields living on the fixed manifold $M$, which represents the spatial container of the continuum. We will fluently switch from a covariant coordinate-free formulation to one which expresses the tensorial quantities in local coordinates, and refer to \cite{Frankel2011TheIntroduction} and \cite{Kanso2007OnMechanics} for a complete derivation of the equations and discussion about the different formulations involving stress. 

The state variables of a continuum mechanical system are the \textit{mass density} top-form $\mu \in \Omega^n(M)$ (where $\mu=\rho \calV$, being $\rho \in \Omega^0(M)$ the scalar density function and $\calV=\star 1 \in \Omega^n(M)$ the metric-induced volume form) and the Eulerian
velocity vector field of the continuum $u\in \vecf{M}$.
The dynamic equation on $\mu$ 
states conservation of mass and reads

\begin{equation}
\label{eq:masscont}
   \dot{\mu}+\lie{u}{\mu}=0.
\end{equation}

The balance of linear momentum constitutes the other dynamic equation of motion. Expressed in the variable $u$, and excluding for simplicity body forces acting on the continuum, it reads:

\begin{equation}
\label{eq:velocityeq}
    \dot{u}+\nabla_u u=\frac{\tstress^{\#}}{\rho},
\end{equation}
where $\tstress \in \Omega^1(M)$ is the force field representation of the effect of any form of stress present in the continuum. The latter equals the divergence of the \textit{Cauchy stress tensor} $\sigma$, which is symmetric due to the balance of angular momentum. 
We choose a tensorial representation for $\sigma$ as a 2-covariant tensor, i.e., in local coordinates $\sigma=\sigma_{ab} \extd x^a \tens \extd x^b$. In this way the traction force over a surface $S\subset M$ characterised by its unit outward normal vector $n=n^i \frac{\partial}{\partial x^i}$ possesses a metric-independent, covector (force-like) nature $\sigma(n,\cdot)=\sigma_{ab}n^b \extd x^a\in \Omega^1(M)$. In this context the stress power over $S$ is calculated as 
$P_{S}=\int_{S} \sigma(n,u) \extd S$. To make tensorial sense of the divergence operation, (either) one of the two indices of $\sigma$ needs to be raised with the inverse metric tensor, i.e., $\tstress=\textrm{div}(\sigma^{\#_2})$, where, in local coordinates, $\sigma^{\#_2}=\sigma_{am}g^{mb} \extd x^a \tens \frac{\partial}{\partial x^b}$, i.e., $\sigma$ is treated as a covector-valued 1-form, and the sharp operator $\#$ is applied to its second leg. The divergence of a $(1,1)$-tensor field $G$ with components $G^b_{a}$ is then defined by means of the usual covariant differentiation of tensor components together with a contraction, i.e.,  $\textrm{div}(G):=\nabla_b G^{b}_{a} \extd x^a \in \Omega^1(M)$, where $\nabla_c G^{b}_{a}:=\partial_c (G^{b}_{a}) + G^{r}_{a} \Gamma^b_{cr} - G^b_{r} \Gamma^{r}_{ca}$.

In the following we will use a formulation for (\ref{eq:velocityeq}) based on Cartan's calculus, for which the stress $\mathcal{T}$ is defined as a covector-valued $(n-1)$-form, i.e., $\mathcal{T}\in \Omega^1(M) \tens \Omega^{n-1}(M)$, and related to Cauchy stress tensor by $\mathcal{T}=\star_2 \sigma$. Such a representation presents different advantages whose complete discussion is out of the scope of this work, and can be found in Appendix A.b of\cite{Frankel2011TheIntroduction}, and in\cite{Gilbert2019AMechanics,Kanso2007OnMechanics,califano2021energetic}.
In local coordinates, we have for $n=3$ that $\mathcal{T} = \frac{1}{2} \mathcal{T}_{abc} \extd x^a \tens \extd x^b \wedge \extd x^c$ with $\mathcal{T}_{abc} = g^{ml} \sigma_{la} \mathcal{V}_{mbc}$.

A technical advantage of this formulation which is directly useful in the sequel of our analysis is the fact that one has a metric-independent notion of \textit{stress power flux form}, expressed as $\stressFlux \in \Omega^{n-1}(M)$.
The $\dot{\wedge}$ binary operator implements i) a duality pairing on the first leg of the tensors; and ii) a regular wedge $\wedge$ at the second leg.
In the specific case of $\stressFlux$, it takes as input $u$ as a vector-valued 0-form and the covector-valued $n-1$ form $\mathcal{T}$ to produce an $n-1$ form which can be integrated over any surface $S$ in $M$ to get the mechanical stress power $P_S=\int_S \tr{\stressFlux}$.

The explicit relation between the stress field $\tstress$ in (\ref{eq:velocityeq}) and $\mathcal{T}$, involves the \textit{exterior covariant derivative} operator $\extcovd: \Omega^1(M) \tens \Omega^{k}(M) \to \Omega^1(M) \tens \Omega^{k+1}(M)$, whose explicit definition can be found in e.g., Chapter 9.3 in\cite{Frankel2011TheIntroduction}. This operator, which generalises the covariant derivative for tensor-valued forms, and that we treat in this work restricted to the case of interest, implements for $\mathcal{T}$ what the divergence operator implements for $\sigma$. In particular it holds \cite{Kanso2007OnMechanics}
\begin{equation}
   \tstress = \star_2 \extcovd \mathcal{T}=\star_2(\textrm{div}(\sigma^{\#_2}) \tens \mathcal{V}) = \textrm{div}(\sigma^{\#_2}).
\end{equation}

In order to write the momentum balance equation ($\ref{eq:velocityeq}$) in a convenient way
we introduce the velocity 1-form $\nu:=u^{\flat}$. Applying the $\flat$ operator to (\ref{eq:velocityeq}), and using the identity (see Chapter 4 in\cite{Arnold1992} for a proof):
\begin{equation*}
    (\nabla_u u)^{\flat}=\lie{u}{\nu}-\frac{1}{2}\extd{\iota_{u}\nu},
\end{equation*}
one obtains
\begin{equation}
\label{eq:velocityeq1}
    \dot{\nu}+\lie{u}{\nu}-\frac{1}{2}\extd{\iota_{u}\nu}=\frac{\tstress}{\rho},
\end{equation}
which is very convenient for calculation purposes since no covariant derivative appears in the expression. This fact allows to express the equation using exterior forms and the use of Cartan's calculus.
The final version of the momentum equation in these variables is given by
\begin{equation}
\label{eq:momentum}
    \dot{\nu}+\lie{u}{\nu}-\frac{1}{2}\extd{\iota_{u}\nu}=\frac{\star_2 \extcovd \mathcal{T}}{\rho}.
\end{equation}

A key property of the exterior covariant derivative $\extcovd$ which will be used later is the Leibniz rule:
\begin{equation}
	\label{eq:leibnizextcovd}
	\extcovd(\stressFlux) = \extcovd u \dot{\wedge}\mathcal{T} + u \dot{\wedge} \extcovd\mathcal{T}.
\end{equation}
Using the fact that $\extcovd$ of tensor-valued 0-forms is the covariant differential (i.e., $\extcovd u=\nabla u$), and that $\extcovd$ of (scalar-valued) forms is the exterior derivative (i.e., $\extcovd{(u \dot{\wedge} \mathcal{T})}=\extd{(u \dot{\wedge} \mathcal{T})}$), we can rewrite the mechanical stress power over the domain $V(t)$ as
\begin{equation}\label{eq:power_balance}
	\int_{\partial V(t)} \tr{\stressFlux} = \int_{V(t)} \extd{(\stressFlux)} = \int_{V(t)} \nabla u \dot{\wedge}\mathcal{T} + u \dot{\wedge} \extcovd\mathcal{T}.
\end{equation}
Here the \textit{velocity gradient} $\nabla u \in \vecf{M} \tens \Omega^1(M)$ is represented in local coordinates as $\nabla u=\nabla_b u^a \frac{\partial}{\partial x^a} \tens \extd x^b:=(\partial_b u^a +u^r \Gamma^a_{br})\frac{\partial}{\partial x^a} \tens \extd x^b$. Notice that the velocity gradient is geometrically not a gradient, but a vector-valued 1-form defined through the covariant derivative such that $\nabla u (X) = \nabla_X u \in \vecf{M}$ for all vector fields $X \in \vecf{M}$, and its terminology is purely based on convention.

Using $\mathcal{T}=\star_2 \sigma$ and $\extcovd{\mathcal{T}}=\star_2 \textrm{div}(\sigma^{\#_2})=\star_2 \tstress=\tstress \tens \mathcal{V}$, the previous power balance can be rewritten as
\begin{align*}
	&\int_{\partial V(t)} \tr{u \dot{\wedge} \star_2 \sigma} 
	= \int_{V(t)} \nabla u \dot{\wedge} \star_2 \sigma + u \dot{\wedge} \star_2 \tstress = \\
	&= \int_{V(t)} \nabla u \dot{\wedge} \star_2 \sigma + \langle \tstress | u \rangle \mathcal{V},
\end{align*}
with $\langle \cdot | \cdot\rangle$ denoting the duality pairing between a covector-field and a vector-field. 

The following proposition is stated to relate the term $\nabla u \dot{\wedge} \star_2 \sigma$ to
the more familiar double contraction of 2-covariant tensors normally encountered using vector calculus on Euclidean spaces.

\begin{proposition}\label{prop_inner}
	Let $\zeta,\sigma \in T_2^0(M) \cong\Omega^1(M)\tens\Omega^1(M)$ be any  (0,2)-tensor fields on $M$. Then the following identity holds:
	$$\zeta^{\#_1} \dot{\wedge} \star_2 \sigma = (\zeta :\sigma) \calV \in \Omega^{n}(M), $$
	where the double dots denote the contraction of (0,2) tensor fields given by the function $$\varepsilon :\sigma = g^{ia}g^{jb}\zeta_{ij}\sigma_{ab} \in \Omega^0(M).$$
\end{proposition}
\begin{proof}
	For the case $n =2$, we have in local coordinates that $\zeta = \zeta_{ij} \extd{x^i}\tens\extd{x^j}$ and $\star_2 \sigma = g^{ml}\sigma_{li}\calV_{mj}\extd{x^i}\tens\extd{x^j}$.
	Let $\gamma:= \zeta^{\#_1} \dot{\wedge} \star_2 \sigma = \frac{1}{2} \gamma_{ab}\extd{x^a}\wedge\extd{x^b}$.
	Since the action of the wedge-dot is a contraction of the first legs and wedge product on the second legs we have that 
	$$\gamma_{ab} = 2 g^{im} g^{kl}\sigma_{li}\zeta_{m\bar{a}}\calV_{k\bar{b}} = g^{im} g^{kl}\sigma_{li} (\zeta_{ma}\calV_{kb} - \zeta_{mb}\calV_{ka}),$$ 
	where we denote by the bars over $a$ and $b$ anti-symmetrization of indices.
	For $(a,b) = (1,2)$ and using the anti-symmetry of $\calV$ we have that 
	\begin{align*}
		\gamma_{12} =& g^{im}\sigma_{li} (g^{kl}\zeta_{m1}\calV_{k2} - g^{kl}\zeta_{m2}\calV_{k1}) \\ 
		=& g^{im} \sigma_{li} (g^{1l}\zeta_{m1}\calV_{12} - g^{2l}\zeta_{m2}\calV_{21}) \\
		=& g^{im} \sigma_{li}\calV_{12}(g^{1l}\varepsilon_{m1} + g^{2l}\zeta_{m2}) = g^{im} \sigma_{li}g^{kl}\zeta_{mk}\calV_{12}.
	\end{align*}
	Similarly, one can show that $\gamma_{21}= g^{im} \sigma_{li}g^{kl}\zeta_{mk}\calV_{21}$. Thus, we have that $\gamma_{ab} = g^{im} \sigma_{li}g^{kl}\zeta_{mk}\calV_{ab}$ and consequently
	$$\gamma= \zeta^{\#_1} \dot{\wedge} \star_2 \sigma = g^{im} \sigma_{li}g^{kl}\zeta_{mk} (\frac{1}{2} \calV_{ab}\extd{x^a}\wedge\extd{x^b}) = (\zeta :\sigma) \calV.$$
	
	For the case $n =3$, the same line of thought holds with $\gamma$ locally expressed as $\frac{1}{3!} \gamma_{abc}\extd{x^a}\wedge\extd{x^b}\wedge\extd{x^c}$ with 
	$$\gamma_{abc} = 3 g^{im} g^{kl}\sigma_{li}\zeta_{m\bar{a}}\calV_{k\bar{b}\bar{c}} =g^{im} g^{kl}\sigma_{li}\zeta_{mk}\calV_{abc}.$$
\end{proof}

Using this result it holds
\begin{equation*}
    \int_{V(t)} \nabla u \dot{\wedge} \star_2 \sigma=\int_{V(t)} ((\nabla u)^{\flat_1}:\sigma)\mathcal{V}.
\end{equation*}
Now we use the property $(\nabla u)^{\flat_1}=\nabla \nu \in \Omega^1(M) \tens \Omega^1(M)$ to write the power balance (\ref{eq:power_balance}) in the more common integration by parts formula:
\begin{align*}
	&\int_{\partial V(t)} \sigma(n,u) \extd S = \int_{V(t)} [(\nabla \nu : \sigma) + \langle\tstress|u\rangle ] \calV
\end{align*}

To conclude this section, we can further massage the power balance $(\ref{eq:power_balance})$ using the fact that the Cauchy stress tensor $\sigma$ is symmetric, a condition which follows from conservation of angular momentum. In particular, for any $\zeta \in T_2^0(M) \cong\Omega^1(M)\tens\Omega^1(M)$, if $\sigma$ is symmetric, it holds (see Chapter 4 in \cite{Schutz1980GeometricalPhysics})
\begin{equation}
    \zeta:\sigma=\textrm{sym}(\zeta):\sigma,
    \label{eq:symmetry}
\end{equation}
where $\textrm{sym}(\cdot)$ returns the symmetric part of its argument.
It follows that (\ref{eq:power_balance}) can be rewritten as
\begin{align*}
	&\int_{\partial V(t)} \tr{u \dot{\wedge} \star_2 \sigma} 
	=  \int_{V(t)} \varepsilon^{\#_1} \dot{\wedge} \star_2 \sigma + \langle \tstress | u\rangle \mathcal{V},
\end{align*}
where $\varepsilon := \textrm{sym}(\nabla \nu)  \in \Omega^1(
M)\tens\Omega^1(M)$.

In words, the dependence of the stress power on the symmetric part only of $\nabla u$ is a direct consequence of the symmetry of $\sigma$.
The tensor $\varepsilon$, can be shown \cite{Gilbert2019AMechanics} to be equal to:
\begin{equation}
\label{eq:rateofstrain}
\varepsilon = \frac{1}{2}\mathcal{L}_u g \in \Omega^1(M)\tens\Omega^1(M)
\end{equation}
which 
extends the concept of \textit{rate of strain} to Riemannian manifolds. For example, its components in a Euclidean space on a Cartesian chart (where all $\Gamma_{ij}^{k}$ vanish) are $(\mathcal{L}_u g)_{ij}=\partial_j u^i+\partial_i u^j$, clearly resembling the standard vector calculus definition of rate of strain in Euclidean space. 
For completeness, we report the full decomposition of the covariant version of the velocity gradient in its symmetric and anti-symmetric components \cite{Gilbert2019AMechanics}
\begin{equation}
    \nabla \nu=\varepsilon-\frac{1}{2} \extd \nu,
\end{equation}
through which one clearly sees that the power balance (\ref{eq:power_balance}) is not affected by the anti-symmetric part of the velocity gradient, which is proportional to the vorticity $2$-form $\extd \nu$.

\section{Thermodynamics}
\label{sec:thermodynamics}

\subsection{First principle}
The first principle of thermodynamics applied to a continuum assumes the existence of a total energy functional over the spatial domain where the continuum evolves, which is conserved at an integral balance level. More specifically, we indicate with $V(t)$ a volume which follows the material domain of the continuum, over which the total energy is $E=\int_{V(t)}\mathcal{E}$, where $\mathcal{E}:\mathcal{X}\to \Omega^n(V)$, maps the state variables $\xi_i \in \mathcal{X}$ of the continuum to the top-form encoding the total energy density field. In order to postulate the first principle in this geometric setting a mathematical expression for the rate of change of $E$ is needed. In particular we fall in the time-varying case discussed in subsection \ref{sec:tvc}, where the integrand $\mathcal{E}$ is a time-dependent top-form (through the time dependency of the state variables $\xi_i$) and the domain of integration $V(t)$ follows a motion which is not a flow, since its generating velocity vector field $u$ is the macroscopic velocity of the continuum, which is in general time-dependent. We thus use (\ref{eq:ReynoldsNonAutonomous}) to express $\dot{E}$, and postulate the first principle in the form
\begin{equation}
\label{eq:firstPrinciple}
    \dot{E}\stackrel{(\ref{eq:ReynoldsNonAutonomous})}{=}\int_{V(t)} \dot{\mathcal{E}}+\lie{u}{\mathcal{E}}\stackrel{\textrm{1st principle}}{=}\int_{ \partial V(t)}\tr{\stressFlux-Q}.
\end{equation}
The equation, which yields the units of power, encodes the facts that variation of total energy in a volume {\it following the continuum} can happen only due to advection of the energy at the boundary, i.e., the total energy does not have sinks or sources within the domain (excluding external body forces and chemical reactions
which we ignore in this work for simplicity of exposition).
The two mechanisms through which energy can be advected at the boundary are the external work due to mechanical stresses $\stressFlux$ and the heat flux $Q$, both $(n-1)$-forms which can be integrated at the boundary $\partial V(t)$ to produce the power flux towards the domain. The trace operator in the integral is formally needed to transform its arguments from $(n-1)$-forms on $V(t)$ to $(n-1)$-forms on $\partial V(t)$, which makes the right term in (\ref{eq:firstPrinciple}) well-defined.
Applying Stokes theorem to the latter expression, and using arbitrariness of the volume $V(t)$ yields the differential energy equation of the continuum
\begin{equation}
\dot{\mathcal{E}}+\lie{u}{\mathcal{E}}=\extd{(\stressFlux-Q)}.
\end{equation}
Using the Leibniz rule (\ref{eq:leibnizextcovd}) for the stress power $\stressFlux$, one obtains
\begin{equation}
\label{eq:totalenergy}
\dot{\mathcal{E}}+\lie{u}{\mathcal{E}}= u \dot{\wedge}\extcovd \mathcal{T}+ \nabla u \dot{\wedge} \mathcal{T} -\extd{Q}.
\end{equation}
Now we use the assumptions that the total 
energy density can be decomposed in $\mathcal{E}=\mathcal{K}+\mathcal{U}$, i.e., it is given by the sum of kinetic energy density $\mathcal{K}(u,\mu)=\frac{1}{2}g(u,u)\mu$ and an internal energy $\mathcal{U}$, which is a macroscopic reflection of several molecular phenomena whose dynamical equation is needed. In order to derive it, we calculate $\dot{\mathcal{K}}$ using continuity and momentum equations (\ref{eq:masscont}) and (\ref{eq:momentum}). One way to perform the calculation is by considering the kinetic energy density a function of the velocity one-form $\nu=u^{\flat}$ and the mass density $\mu$, i.e., $\mathcal{K}=\frac{1}{2}(\iota_{\nu^{\#}}\nu) \mu$, and computing the rate using the chain rule (see e.g., \cite{califano2021geometric,VanDerSchaft2002HamiltonianFlow}) as
$\dot{\mathcal{K}}=\frac{1}{2}(\iota_{\nu^{\#}}\nu)\wedge \dot{\mu}+(\iota_{\nu^{\#}}\mu) \wedge \dot{\nu}$. 
Using the field equations (\ref{eq:masscont},\ref{eq:momentum}), one obtains
\begin{equation}
    \dot{\mathcal{K}}=-\extd{} \left(\frac{1}{2}i_{\nu^{\#}}\nu \wedge i_{\nu^{\#}}\mu \right)+\iota_{\nu^{\#}}(\extcovd \mathcal{T}).
\end{equation}
Recognising the term inside the parenthesis $\frac{1}{2}i_{\nu^{\#}}\nu \wedge i_{\nu^{\#}}\mu=\iota_{\nu^{\#}} \mathcal{K}$, returning to use $u$ for $\nu^{\#}$, and applying Cartan's formula, we obtain the kinetic energy density differential equation
\begin{equation}
\label{eq:kineticenergy}
    \dot{\mathcal{K}}+\lie{u}{\mathcal{K}}=u \dot{\wedge}\extcovd \mathcal{T}.
\end{equation}
In order to obtain an equation in the internal energy density we simply subtract (\ref{eq:kineticenergy}) from (\ref{eq:totalenergy}), obtaining

\begin{equation}
\label{eq:internalenergy}
    \dot{\mathcal{U}}+\lie{u}{\mathcal{U}}= \nabla u \dot{\wedge} \mathcal{T}-\extd{Q},
\end{equation}
rightfully indicating that the stress power $\nabla u \dot{\wedge} \mathcal{T}$ is a source for internal energy.

\subsection{Second principle}

In a state of thermodynamic equilibrium it is assumed that the so-called \textit{Gibbs relation} is valid, relating entropy and other thermodynamic potentials as
\begin{equation}
\label{eq:gibbsnonrigorous}
    T \delta \calS= \delta \mathcal{U} +p \delta \calV, 
\end{equation}
where $\calS$ is the entropy, $\calV$ is the volume, $T$ is the temperature, and $p$ is the pressure. For simplicity of exposition we ignore the component of Gibbs relation corresponding to reactive flows in the following. 
We should here clarify the role of the "differential" $\delta$ and the nature of the thermodynamic variables involved in this formula in order to \textit{postulate} a
version of (\ref{eq:gibbsnonrigorous}) which possesses mathematical rigour consistent with the differential geometric description that we use and 
leads to a correct entropy equation, i.e., a dynamic equation encoding the second principle of thermodynamics. 

As first consideration, entropy and volume are \textit{extensive} thermodynamic variables, which translates in this differential geometric context to require them being densities (top-forms) over the volume where the continuum evolves, i.e., $\calS,\calV \in \Omega^n(V)$. Temperature and pressure are instead \textit{intensive} thermodynamic variables (i.e, variables which can be evaluated at a point in space, but do not possess additive property) which translates to require them being $0$-forms, i.e., $T,p \in \Omega^0(V)$. Notice that this assignments make the exterior derivative $\extd{}$ a bad candidate for the differential $\delta$, since (\ref{eq:gibbsnonrigorous}) would produce a collapsing identity on $n+1$ forms.

Secondly, we are dealing with non-equilibrium field equations describing the dynamic evolution of a continuum, and as such we need to \textit{assume} that a relation in the form (\ref{eq:gibbsnonrigorous}) holds also in this condition, where a non null macroscopic velocity field $u$ is one of the state variables of the system.

Third, the Gibbs relation needs to be invariant with respect to Galilean transformations, that is it must look the same for all inertial observers.

A natural choice in this context for the differential $\delta$ in (\ref{eq:gibbsnonrigorous}) is the \textit{material derivative} operator $\DDt{}:=\frac{\partial}{\partial t}+\lie{u}{}$, which computes the variation of its argument along the streamline generated by $u$. 
Thus we postulate the following version of Gibbs relation:
\begin{equation}
\label{eq:gibbsrigorous}
    T \DDt{\calS} = \DDt{\mathcal{U}} +p \DDt{\calV}, 
\end{equation}
where we use for $\calV$ the natural volume form induced by the Riemannian metric, i.e., $\calV=\star 1$.

\begin{remark}
	Notice that the version (\ref{eq:gibbsrigorous}) of Gibbs relation, postulated to be consistent in the space of top-forms, is equivalent to the more frequently proposed vector calculus version (see e.g., \cite{Asinari2016OverviewDynamics}) which involves states of specific internal energy $U  \in \Omega^0(V)$, specific entropy $S \in \Omega^0(V)$ and specific volume $1/\rho \in \Omega^0(V)$. To see this we first relate the specific ($0$-forms) and density (top-forms) quantities by $\mathcal{U}=U \wedge \mu$ and $\mathcal{S}=S \wedge \mu$. Since by (\ref{eq:masscont}) $\DDt{\mu}=0$,  (\ref{eq:gibbsrigorous}) can be rewritten as
	
	\begin{equation*}
		T \DDt{S}\wedge \mu = \DDt{U} \wedge \mu +p \DDt{\calV},
	\end{equation*}
	but $\mathcal{V}=\star 1=\frac{1}{\rho} \wedge \mu$, and as a consequence
	\begin{equation*}
		T \DDt{S}= \DDt{U} +p \DDt{} \left( \frac{1}{\rho} \right)
	\end{equation*}
	holds on the space of scalar functions, where the material derivative is equivalent to the vector calculus version $\DDt{}:=\frac{\partial}{\partial t}+ u \cdot \nabla$, where the ``nabla'' operator is meant in the vector calculus sense, as the gradient operator acting on scalar functions.
\end{remark}

Now we show that the choice (\ref{eq:gibbsrigorous}) consistently leads to a representation of the second principle in this differential geometric context. Substituting (\ref{eq:internalenergy}) in (\ref{eq:gibbsrigorous}) yields
\begin{equation}\label{eq:Entrpy_intermed}
    T \DDt{\calS} = \nabla u \dot{\wedge}\mathcal{T} -\extd{Q} +p \DDt{\calV}.
\end{equation}
In order to proceed from here we must use a standard assumption involving the constitutive relation of a continuum regarding the stress form $\mathcal{T}$. 
In particular the stress can be decomposed in the sum of a \textit{pressure stress} $\mathcal{T}_p=- p\mathcal{V} \in \Omega^1(M)\tens \Omega^{n-1}(M) $ which is solely determined by the pressure function in (\ref{eq:gibbsrigorous}), and the residual stress $\mathcal{T}_r=\mathcal{T}-\mathcal{T}_p$, depending on the constitutive equations of the considered continuum.
An explicit calculation of $ \nabla u \dot{\wedge}\mathcal{T}_p$ can be obtained by the following proposition.

\begin{proposition}\label{prop:nablaU}
	For any $f\in \Omega^0(M)$ and $u\in \vecf{M}$, the pairing $\nabla u \dot{\wedge} f \mathcal{V}$ between 
	$\nabla u \in \vecf{M} \tens \Omega^1(M)$ and $f \mathcal{V} \in \Omega^1(M) \tens \Omega^{n-1}(M)$ is given by
	$$\nabla u \dot{\wedge} f \mathcal{V} = f  {\normalfont\textrm{div}}(u) \calV \in \Omega^{n}(M).$$
\end{proposition}
\begin{proof}
	Using the Leibniz rule for the exterior covariant derivative (\ref{eq:leibnizextcovd}) followed by the Leibniz rules of the exterior derivative and the interior product, we get
	\begin{align*}
		\nabla u \dot{\wedge} f \mathcal{V} &= \extd{(u\dot{\wedge} f \mathcal{V})} - u \dot{\wedge} \extcovd(f \mathcal{V}) = \extd{(f \wedge \iota_{u}\mathcal{V})} - u \dot{\wedge} (\extd{f}\tens \mathcal{V}) \\
		&= \extd{f} \wedge \iota_{u}\mathcal{V} + f \wedge \extd{}\iota_{u}\mathcal{V} - \iota_u\extd{f}\wedge \mathcal{V} \\
		&= \extd{f} \wedge \iota_{u}\mathcal{V} + f \wedge \extd{}\iota_{u}\mathcal{V} -\extd{f}\wedge \iota_u\mathcal{V} = f \wedge \extd{}\iota_{u}\mathcal{V}.
	\end{align*}
	Using $\extd{}\iota_{u}\mathcal{V} =  \lie{u}{\calV} = \textrm{div}(u) \calV$ concludes the proof.
\end{proof}

Consequently, we have that $\nabla u \dot{\wedge}\mathcal{T}_p = - p \textrm{div}(u) \calV$.
Notice that this term is proportional to the divergence of the velocity vector field of the continuum, and appears with opposite sign in the kinetic and internal energy equations, representing a reversible exchange between these two energy fields due to work of compression or expansion.
Using this result in (\ref{eq:Entrpy_intermed}) and the fact that $\DDt{\calV}=\textrm{div}(u) \calV$, we end up with the entropy equation
\begin{equation}
\label{eq:entropy}
    \DDt{\calS}= \frac{1}{T}\nabla u \dot{\wedge}\mathcal{T}_r -\frac{\extd{Q}}{T},
\end{equation}
which serves as a definition for the extensive entropy field $\calS$ of the continuum.
An insightful form of the latter equation is obtained applying the integration by parts identity \[\frac{\extd Q}{T}=\extd{} \left(\frac{Q}{T} \right)+\frac{1}{T^2}  \extd{T} \wedge Q.\]
Substituting the latter in (\ref{eq:entropy}) one obtains 

\begin{equation}
\label{eq:entropy2}
\dot{\mathcal{S}}+\mathcal{L}_{u} \mathcal{S}  =\tentropy_T+\tentropy_{\mathcal{T}_r}-\extd{} \left(\frac{Q}{T} \right),
\end{equation}
with 
\begin{equation*}
	\tentropy_T:=- \frac{1}{T^2} \extd{T} \wedge Q, \qquad \tentropy_{\mathcal{T}_r}:=\frac{1}{T} \nabla u \dot{\wedge}\mathcal{T}_r.
\end{equation*}
Integrating the latter on a volume $V(t)$ which follows the continuum, we recognise (by means of (\ref{eq:ReynoldsNonAutonomous})) on the left hand side the variation of the total entropy $S_V:=\int_{V(t)} \mathcal{S}$, i.e.,
\begin{equation}
    \dot{S}_V=\int_{V(t)} \tentropy_T+\tentropy_{\mathcal{T}_r} -\int_{\partial V(t)} \textrm{tr} \left( \frac{Q}{T} \right),
\end{equation}
which is an integral equation representing a balance law like the one postulated for the total energy in (\ref{eq:totalenergy}), but this time with source terms.
In fact $\tentropy_T$ is the entropy source term due to non null temperature gradients, while $\tentropy_{\mathcal{T}_r}$ is the entropy source due to non conservative stresses.

The constitutive equations defining $\mathcal{T}_r$ and $Q$ for the specific material should be constraint in a way that $\tentropy_{T} \geq 0$ and $\tentropy_{\mathcal{T}_r} \geq 0$ hold true, which is a statement for the second principle of thermodynamics (entropy can have only positive source terms). 
Notice that a boundary term of entropy is present in case of non-adiabatic flows (i.e., $\tr{Q} \neq 0$). In case of adiabatic flow, or flow over a closed manifold (i.e., $\partial V =\emptyset$), one gets the standard $\dot{S}_V \geq 0$, i.e., total entropy of a closed system cannot decrease autonomously.

\section{Fourier-Navier-Stokes fluids}
\label{sec:FNS}
Now we apply the general formulation of the previous section to the case of Fourier-Navier-Stokes fluids by specifying the constitutive equations relating $\mathcal{T}_r$ and $Q$ to the continuum states $u, \mu$ and $T$ for Newtonian fluids.

For the heat flux $Q$ we chose the simple linear constitutive equation involving the \textit{heat conductivity coefficient} $k$, also called \textit{Fourier law}:
\begin{equation}
	\label{eq:Qconstitutive}
	Q=-k \star \extd T.
\end{equation}
With this choice one easily verifies that the entropy source term due to non null temperature gradients $$\tentropy_T=\frac{k}{T^2} \extd T \wedge \star \extd T\geq 0$$ is non-negative.

Whereas for the residual stress tensor $\mathcal{T}_r$, it is given by the \textit{viscous stress tensor}, given by the sum of a bulk stress $\mathcal{T}_{\lambda}$ and a shear stress $\mathcal{T}_{\kappa}$, defined by \cite{Gilbert2019AMechanics,califano2021geometric}
\begin{align}
    \mathcal{T}_{\lambda} &:= \lambda (\textrm{div}(u)\calV) ,
    \label{eq:DefBulkViscosity}\\
   \mathcal{T}_{\kappa} &:= 
2\kappa (\star_2 \epsilon), \label{eq:DefShearViscosity}
\end{align}
where $\lambda$ and $\kappa$ are respectively the bulk and shear viscosity coefficients, and $\varepsilon$ is given by (\ref{eq:rateofstrain}).
The bulk stress $\mathcal{T}_{\lambda}$ depends on the divergence of the velocity vector field $\textrm{div}(u)$, while the shear stress $\mathcal{T}_{\kappa}$ is defined in order to model viscous stresses whenever the transport of the metric under the flow of $u$ is non-zero, i.e., when $u$ fails to be the generator of a rigid body motion.

Using Props. \ref{prop_inner}, \ref{prop:nablaU} and equation (\ref{eq:symmetry}), the entropy source term $\tentropy_{\mathcal{T}_r}$ due to the viscous stress in (\ref{eq:entropy2}) is expressed as:
\begin{align*}
	\tentropy_{\mathcal{T}_r} = \frac{1}{T} \nabla u \dot{\wedge}\mathcal{T}_r
							&= \frac{\lambda}{T} \nabla u \dot{\wedge} \textrm{div}(u)\calV + \frac{2\kappa}{T} \nabla u \dot{\wedge} \star_2 \epsilon \\
							&= \frac{\lambda}{T} \textrm{div}(u) \wedge \star \textrm{div}(u) + \frac{2\kappa}{T} \text{sym}(\nabla u) \dot{\wedge} \star_2 \epsilon \\
							&= \frac{\lambda}{T} \textrm{div}(u) \wedge \star \textrm{div}(u) + \frac{2\kappa}{T} \epsilon^{\#_1} \dot{\wedge} \star_2 \epsilon,\\
							&= \frac{1}{T}(\lambda \textrm{div}(u)^2 + 2\kappa \epsilon:\epsilon) \calV \geq 0,
\end{align*}
which results in a non-negative source of entropy driven by non-null divergence and rate of strain of the macroscopic vector field $u$.

We finally stress that the positiveness of the entropy source terms follow in this setting by the choice of Newtonian constitutive relations for the stress and $k>0$, but in case of non-Newtonian fluids it represents a \textit{constraint} for the constitutive relations that can be used, since a negativeness of those terms would constitute a violation of the second principle.

\section{Conclusion}
\label{sec:conclusion}
An overview of entropy creation mechanism in the context of continuum mechanics on general manifolds has been presented. The specialisation to Fourier-Navier-Stokes flows shows the consistency of the generalisation that this axiomatic treatment provides to vector-calculus formulations valid in Euclidean spaces. Future work aims at investigating the geometric structure of infinite-dimensional thermodynamic phase-space, potentially leading to novel structure-preserving discretisation schemes and control strategies.

\section*{Funding}
This work was supported by the PortWings project funded by the European Research Council [Grant Agreement No. 787675]

\bibliography{references1}

\end{document}

%% file: commands.tex
\newcommand{\vecf}[1]{\mathfrak X({#1})}

\newcommand{\ddt}[1]{\textrm{d}{#1}/ \textrm{dt}}

\newcommand{\DDt}[1]{\frac{\textrm{D} {#1}}{\textrm{Dt}}}

\newcommand{\lie}[2]{\mathcal{L}_{#1}{#2}}

\newcommand{\extd}[1]{\textrm{d}{#1}}

\newcommand{\tr}[1]{\textrm{tr}({#1})}

\newcommand{\extcovd}{\extd{}_{\nabla}}

\newcommand{\calS}{\mathcal{S}}

\newcommand{\calV}{\mathcal{V}}

\newcommand{\tens}{\otimes}

\newcommand{\tstress}{\textrm{t}}

\newcommand{\tentropy}{\textfrak{d}}